%% file: OnTheC0InextendibilityFinal.tex
\documentclass[reqno]{article}

\usepackage{amsthm}
\usepackage{amsmath}
\usepackage{amssymb}
\usepackage{graphicx}
\usepackage{color}
\usepackage{bbm}
\usepackage[top=2.5cm, bottom=2.5cm, left=3cm, right=3cm]{geometry}

\usepackage{enumerate}
\usepackage{faktor}
\usepackage{nicefrac}

\newcommand{\N}{\mathbb{N}}
\newcommand{\R}{\mathbb{R}}

\newcommand{\ed}{R_{\varepsilon_0, \varepsilon_1}}

\newcommand{\ux}{\underline{x}}
\newcommand{\Sd}{\mathbb{S}^{d-1}}
\newcommand{\Mint}{M_{\mathrm{int}}}
\newcommand{\gint}{g_{\mathrm{int}}}

\newcommand{\Mmax}{M_{\mathrm{max}}}
\newcommand{\gmax}{g_{\mathrm{max}}}

\begin{document}

\numberwithin{equation}{section}
\newtheorem{theorem}[equation]{Theorem}
\newtheorem{remark}[equation]{Remark}
\newtheorem{claim}[equation]{Claim}
\newtheorem{lemma}[equation]{Lemma}
\newtheorem{definition}[equation]{Definition}
\newtheorem{premiss}[equation]{Premiss}
\newtheorem{corollary}[equation]{Corollary}
\newtheorem{proposition}[equation]{Proposition}
\newtheorem*{theorem*}{Theorem}

\title{On the proof of the $C^0$-inextendibility of the Schwarzschild spacetime}
\author{Jan Sbierski\thanks{Department for Pure Mathematics and Mathematical Statistics, University of Cambridge,
Wilberforce Road,
Cambridge,
CB3 0WA,
United Kingdom}}

\maketitle

\begin{abstract}
This article presents a streamlined version of the author's original proof of the $C^0$-inextendibility of the maximal analytic Schwarzschild spacetime. Firstly, we deviate from the original proof by using the result, recently established in collaboration with Galloway and Ling, that given a $C^0$-extension of a globally hyperbolic spacetime, one can find a timelike geodesic that leaves this spacetime. This result much simplifies the proof of the inextendibility through the exterior region of the Schwarzschild spacetime. Secondly, we give a more flexible and shorter argument for the inextendibility through the interior region. Furthermore, we present a small new structural result for the boundary of a globally hyperbolic spacetime within a $C^0$-extension which serves as a new and simpler starting point for the proof.
\end{abstract}

\section{Introduction}
This note presents a streamlined proof of the following
\begin{theorem}\label{MainThm}
The maximal analytic extension $(\Mmax, \gmax)$ of the Schwarzschild spacetime is $C^0$-inextendible.
\end{theorem}
The original proof of this theorem given in \cite{Sbie15} was divided into three steps: Firstly, one assumes that there is a $C^0$-extension of the Schwarzschild spacetime $(\Mmax, \gmax)$ and then shows that there is a (without loss of generality) future directed and future inextendible timelike curve $\gamma$ in the Schwarzschild spacetime that can be extended to the future as a timelike curve in the $C^0$-extension. Secondly, one shows that the curve $\gamma$ cannot lie in the exterior region of the Schwarzschild black hole and thus leave the spacetime through timelike or null infinity. Here, one uses the future one-connectedness\footnote{Recall that a time-oriented Lorentzian manifold is called \emph{future one-connected} iff any two timelike curves with the same endpoints are homotopic with fixed endpoints \underline{via timelike curves}; see also Definition 2.13 of \cite{Sbie15}.} of the exterior region together with its future divergence\footnote{This is the property, that for any future directed inextendible timelike curve $\gamma : [0,1) \to M$ one has $d(\gamma(0), \gamma(s)) \to \infty $ for $s \to 1$, where $d$ is the Lorentzian distance function.}. Especially the property of future one-connectedness is quite tedious to establish. Thirdly, one shows that the timelike curve $\gamma$ can neither be ultimately contained in the interior of the Schwarzschild black hole and thus leave through the curvature singularity at $r = 0$. Here, one establishes that the \emph{spacelike diameter}\footnote{This is a geometric quantity at the $C^0$-level of the metric.} of a certain spherically symmetric subset `touching' the curvature singularity, is infinite. Covering this spherically symmetric subset by a finite number of charts in which the metric is suitably controlled, and which are obtained from the assumed $C^0$-extension, one obtains the contradiction that the spacelike diameter should be finite. For the construction of the finitely many charts one in particular exploits the spherical symmetry of the Schwarzschild spacetime. 

The strategy of the proof given in this note deviates as follows from the original strategy: The first step is replaced by a result established in \cite{GalLinSbi17}, which states that given a $C^0$-extension of a globally hyperbolic Lorentzian manifold $(M,g)$ there exists a (without loss of generality) future directed and future inextendible timelike \emph{geodesic} in $M$ which has a future limit point on the boundary of $M$ in the extension. Applying this result to the maximal analytic Schwarzschild spacetime, one again distinguishes the two cases of the timelike geodesic leaving through timelike/null infinity or through the curvature singularity at $r=0$. The first case is now easily ruled out using the completeness of geodesics approaching timelike or null infinity -- without the need of establishing the future one-connectedness of the exterior. For the second case we present a new, \emph{local} version of the obstruction to $C^0$-extensions through $r=0$, which does not rely on the infinitude of the spacelike diameter of the \emph{global} spherically symmetric set constructed in \cite{Sbie15} (which has to be covered by multiple charts), but transfers the idea locally to a single chart. This new argument is more flexible and can be more easily adapted to other situations; this part of the argument does not rely on spherical symmetry. We should also point out that the proof overall still requires the future one-connectedness of the interior.

For the sake of brevity this note is not self-contained. We will often refer the reader to definitions and results from \cite{Sbie15}. We begin by presenting a  small new structural result for $C^0$-extensions which is convenient to have at one's disposal.

\section{A structural result for $C^0$-extensions}

This section establishes an easy yet fundamental result on the structure of $C^0$-extensions. We implicitly assume that all manifolds are connected. We then recall from Section 2.3 of \cite{Sbie15} that a \emph{$C^0$-extension} of a  smooth Lorentzian manifold $(M,g)$ consists of a  Lorentzian manifold $(\tilde{M}, \tilde{g})$ of the same dimension as $M$, $\tilde{g} \in C^0$, together with an isometric embedding $\iota : M \hookrightarrow \tilde{M}$, such that $\iota(M)$ is a proper subset of $\tilde{M}$. If no such $C^0$-extension exists, we say that $(M,g)$ is \emph{$C^0$-inextendible}. Finally, we recall that it is convenient for our purposes to define a \emph{timelike curve} to be a piecewise smooth curve such that, where defined, the tangent vector is everywhere timelike and at the points of discontinuity, the left and right sided tangent vectors are timelike and lie in the same connectedness component of the timelike cone. We recall from Proposition 2.6 of \cite{Sbie15} that, with this definition of timelike curve, the timelike future and past of a point in a time-oriented Lorentzian manifold remains open even for merely continuous metrics. 

In the following we define the future (past) boundary of a time-oriented Lorentzian manifold with respect to a $C^0$-extension (see also Definition  2.1 in \cite{GalLin16}).
\begin{definition}
Let $(M,g)$ be a time-oriented Lorentzian manifold with a continuous metric and $\iota : M \hookrightarrow \tilde{M}$ a $C^0$-extension of $M$. The \emph{future boundary of $M$} is the set $\partial^+\iota(M) $ consisting of all points $p \in \tilde{M}$ such that there exists a timelike curve $\tilde{\gamma} : [-1,0] \to \tilde{M}$ such that $\mathrm{Im}(\tilde{\gamma}|_{[-1,0)}) \subseteq \iota(M)$, $\tilde{\gamma}(0) = p \in \partial \iota(M)$, and $\iota^{-1} \circ \tilde{\gamma}|_{[-1,0)}$ is future directed in $M$.
\end{definition}
Clearly we have $\partial^+\iota(M) \subseteq \partial \iota(M)$. The past boundary $\partial^- \iota(M)$ is defined analogously. Moreover, it follows from Lemma 2.17 of \cite{Sbie15} that  $\partial^+\iota(M) \cup \partial^- \iota(M) \neq \emptyset$.

Figure \ref{FigPosStr} below gives a first idea of possible boundary structures. In the left Penrose-like diagram the segment from $q$ to $r$, closed at $q$, open at $r$, consists of future boundary points. The segment from $p$ to $r$, closed at $p$, open at $r$, consists of past boundary points. The boundary point $r$ lies neither in the future nor the past boundary. And finally all points of the open segment from $p$ to $q$ lie in the future as well as in the past boundary.

The right diagram shows a $C^0$-extension of a globally hyperbolic Lorentzian manifold. Again, the future and past boundaries are not disjoint. Moreover, the future (past) boundary is not achronal in the extension, nor in any small neighbourhood of the future boundary point $s$.

\begin{figure}[h]
  \centering
  \def\svgwidth{7cm}
    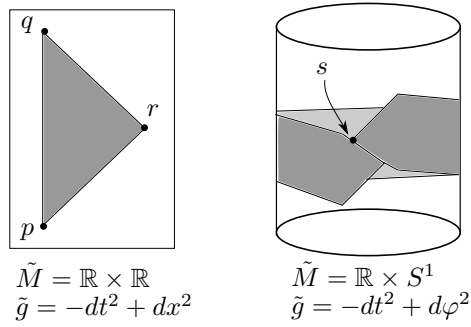
    \caption{Possible structures of the future and past boundaries} \label{FigPosStr}
\end{figure}

However, we have the following result for  globally hyperbolic $(M,g)$.

\begin{proposition}\label{FundProp}
Let $\iota : M \hookrightarrow \tilde{M}$ be a $C^0$-extension of a globally hyperbolic Lorentzian manifold $(M,g)$  and let $p \in \partial^+ \iota(M)$. For every $\delta >0$ there exists a chart $\tilde{\varphi} : \tilde{U} \to(-\varepsilon_0, \varepsilon_0) \times  (-\varepsilon_1, \varepsilon_1)^{d}$, $\varepsilon_0, \varepsilon_1 >0$ with the following properties
\begin{enumerate}[i)]
\item $p \in \tilde{U}$ and $\tilde{\varphi}(p) = (0, \ldots, 0)$
\item $|\tilde{g}_{\mu \nu} - m_{\mu \nu}| < \delta$, where $m_{\mu \nu} = \mathrm{diag}(-1, 1, \ldots , 1)$
\item There exists a Lipschitz continuous function $f : (-\varepsilon_1, \varepsilon_1)^d \to (-\varepsilon_0, \varepsilon_0)$ with the following property: 
\begin{equation}\label{PropF1}
\{(x_0,\underline{x}) \in (-\varepsilon_0, \varepsilon_0) \times (-\varepsilon_1, \varepsilon_1)^{d} \; | \: x_0 < f(\underline{x})\} \subseteq \tilde{\varphi}\big( \iota(M) \cap \tilde{U} \big)
\end{equation} and 
\begin{equation}\label{PropF2}
\{(x_0,\underline{x}) \in (-\varepsilon_0, \varepsilon_0) \times (-\varepsilon_1, \varepsilon_1)^{d}  \; | \: x_0 = f(\underline{x})\} \subseteq \tilde{\varphi}\big(\partial^+\iota(M)\cap \tilde{U}\big) \;.
\end{equation}
Moreover, the set on the left hand side of \eqref{PropF2}, i.e. the graph of $f$, is achronal in $(-\varepsilon_0, \varepsilon_0) \times  (-\varepsilon_1, \varepsilon_1)^{d}$.
\end{enumerate}
\end{proposition}

\begin{remark}\label{RemBoundary}
\begin{enumerate}
\item Note that the achronality of the graph of $f$ in particular entails that any past directed timelike curve starting below the graph of $f$ remains below the graph of $f$ (and thus in $\tilde{\varphi}\big(\iota(M)\cap \tilde{U}\big)$). This will be used often.
\item
As shown by the proof, given a Cauchy hypersurface $\Sigma$ of $M$ it can also be arranged that we have in fact $$\{(x_0,\underline{x}) \in (-\varepsilon_0, \varepsilon_0) \times (-\varepsilon_1, \varepsilon_1)^{d} \; | \: x_0 < f(\underline{x})\} \subseteq \tilde{\varphi} \big( \iota\big(I^+(\Sigma,M)\big)\cap \tilde{U}\big)\;,$$
\end{enumerate}
\end{remark}

\begin{proof}
Since $p \in \partial^+\iota(M)$, there exists a timelike curve $\tilde{\gamma} : [-1, 0] \to \tilde{M}$  with $\mathrm{Im} (\tilde{\gamma}|_{[-1,0)}) \subseteq \iota(M)$ and $\tilde{\gamma}(0)=p \in \partial\bigl(\iota(M)\big)$. After a possible reparametrisation of $\tilde{\gamma}$ there exists  a chart $\tilde{\varphi} : \tilde{U} \to (-\varepsilon_0, \varepsilon_0) \times (-\varepsilon_1, \varepsilon_1)^{d} $ such that (see also \cite{Sbie15}, Lemma 2.4)
\begin{enumerate}[i)]
\item $(\tilde{\varphi} \circ \tilde{\gamma})(s) = (s, 0, \ldots, 0)$ for $s \in (-\varepsilon_0, 0]$
\item $|\tilde{g}_{\mu \nu} - m_{\mu \nu}| < \delta$
\end{enumerate}
In particular $\delta >0$ can be chosen so small that $- \frac{1}{2} dx_0^2 + dx_1^2 + \ldots + dx_d^2 \prec \tilde{g}$, which means that the light cones of the first Lorentzian metric are contained in the light cones of $\tilde{g}$.

Let $\Sigma \subseteq M$ be a Cauchy hypersurface of $M$ and let $\gamma := \iota^{-1} \circ \tilde{\gamma}|_{[-1,0)}$. For $s_0 <0$ close enough to $0$ we have $\gamma(s_0) \in I^+(\Sigma, M)$. Since $I^+(\Sigma, M)$ is an open set, we can choose $\varepsilon_0, \varepsilon_1$ smaller, if necessary, with $-\varepsilon_0 < s_0 < 0$, such that $(-\varepsilon_0,s_0] \times (-\varepsilon_1, \varepsilon_1)^d \subseteq \tilde{\varphi} \big( \iota(I^+(\Sigma, M))\cap \tilde{U}\big)$. Furthermore, we can assume that $\varepsilon_1 < \frac{1}{2} \varepsilon_0$. Together with our choice of $\delta$ this implies in particular that for any point $\underline{x} \in (-\varepsilon_1, \varepsilon_1)^d$ the straight line connecting $(\frac{9}{10}\varepsilon_0, \underline{x})$ with $0$  is timelike.

We now define $f : (-\varepsilon_1, \varepsilon_1)^d \to (-\varepsilon_0, \varepsilon_0]$ by $$f(\underline{x}) = \sup \{ s_0 \in (-\varepsilon_0, \varepsilon_0) \; | \; (s, \underline{x}) \in \tilde{\varphi} \big( \iota(I^+(\Sigma, M))\cap \tilde{U}\big) \quad \forall s \in (-\varepsilon_0, s_0)\} \;.$$
We first show that for all $\underline{x} \in (-\varepsilon_1, \varepsilon_1)^d$ we have $f(\underline{x}) < \varepsilon_0$. 

Assuming the negation, there exists an $\underline{x}_0 \in (-\varepsilon_1, \varepsilon_1)^d$ with  $f(\underline{x}_0) = \varepsilon_0$. It follows that $(\frac{9}{10}\varepsilon_0, \underline{x}_0) \in \tilde{\varphi} \big( \iota(I^+(\Sigma, M))\cap \tilde{U}\big)$. Let $\sigma = (\sigma_0, \underline{\sigma}) : [0,1] \to (-\varepsilon_0, \varepsilon_0) \times (-\varepsilon_1, \varepsilon_1)^d$ denote the past directed timelike straight line from $\sigma(0) = (\frac{9}{10} \varepsilon_0, \underline{x}_0)$ to $\sigma(1) = 0$. We claim that $\sigma|_{[0,1)}$ maps into $\tilde{\varphi} \big( \iota(I^+(\Sigma, M))\cap \tilde{U}\big)$.

To see this we (partially) foliate the plane $$\{ (t,\underline{\sigma}(s)) \in (-\varepsilon_0, \varepsilon_0) \times (-\varepsilon_1, \varepsilon_1)^d \; | \; t \leq \sigma_0(s), s \in (0,1]\}$$ 
by the (closed) straight lines $\rho_\tau$ starting at $(\tau, \underline{x}_0)$ with slope $(\frac{9}{10}\varepsilon_0, -\underline{x}_0)$ and ending on $\sigma$.  It follows from the openness of $\tilde{\varphi} \big(\iota(M)\cap \tilde{U}\big)$ that there is a $-\frac{9}{10}\varepsilon_0 <\tau' < \frac{9}{10}\varepsilon_0$ such that 
\begin{equation}\label{InM}
\rho_{\tau} \subseteq \tilde{\varphi} \big( \iota(M) \cap \tilde{U}\big) \;\textnormal{ holds } \forall \tau \in (\tau', \frac{9}{10}\varepsilon_0)\;.
\end{equation}
Let $\tau_0 \in [-\frac{9}{10}\varepsilon_0, \frac{9}{10}\varepsilon_0)$ be the infimum over all  such $\tau'$ with the property \eqref{InM}. We first remark that for $\tau \in (\tau_0, \frac{9}{10}\varepsilon_0)$ we have $\rho_\tau \subseteq \tilde{\varphi}\big(\iota(I^+(M, \Sigma))\cap \tilde{U}\big)$. Assuming that $\tau_0 > -\frac{9}{10}\varepsilon_0$, there exists a point $q$ on $\rho_{\tau_0}$ with $q \notin \tilde{\varphi} \big( \iota(M)\cap \tilde{U}\big)$. This however allows us to construct a past directed timelike curve in $M$ which is past inextendible and completely contained in $I^+(\Sigma, M)$, see Figure \ref{FigCont}. This is a contradiction to $\Sigma$ being a Cauchy hypersurface of $M$, and thus we obtain $\tau_0 = -\frac{9}{10} \varepsilon_0$. This shows that $\sigma|_{[0,1)}$ maps into $\tilde{\varphi} \big( \iota(I^+(\Sigma, M))\cap \tilde{U}\big)$.

\begin{figure}[h]
  \centering
  \def\svgwidth{5cm}
    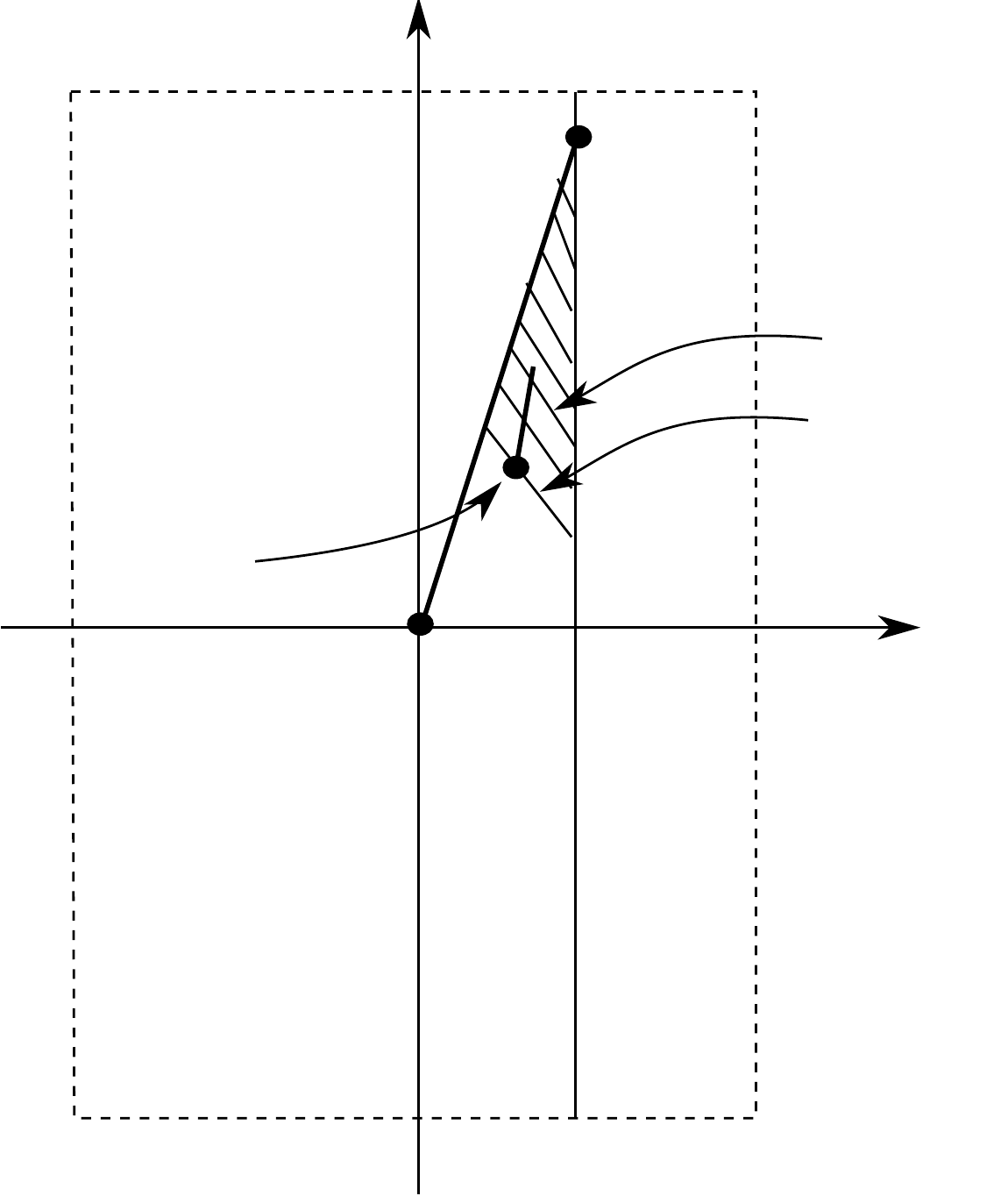
    \caption{The continuity argument} \label{FigCont}
\end{figure}

Again, the curve $\sigma|_{[0,1)}$ corresponds to a past directed timelike curve in $M$ which  is past inextendible and completely contained in $I^+(\Sigma, M)$ -- a contradiction. We thus obtain $f(\underline{x}) < \varepsilon_0$ for all $\underline{x} \in (-\varepsilon_1, \varepsilon_1)^d$. 

The properties \eqref{PropF1} and \eqref{PropF2} are immediate. To show that the function $f$ is continuous we use a similar construction as before: Let $\underline{x}_n \in (-\varepsilon_1, \varepsilon_1)^d$ be a sequence with $\underline{x}_n \to \underline{x}_\infty \in (-\varepsilon_1, \varepsilon_1)^d$. Assume that $f(\underline{x}_n) \not\to f(\underline{x}_\infty)$. Then there exists a $\mu >0$ and a subsequence $n_k$ such that $|f(\underline{x}_{n_k}) - f(\underline{x}_\infty)| > \mu$. After possibly extracting another subsequence we assume $f(\underline{x}_{n_k})  > f(\underline{x}_\infty) + \mu$; the case $f(\underline{x}_{n_k})  < f(\underline{x}_\infty) - \mu$  is dealt with analogously. For $k$ big enough we can connect $(f(\underline{x}_{n_k}) - \frac{\mu}{2}, \underline{x}_{n_k})$ to $(f(\underline{x}_\infty), \underline{x}_\infty)$ by a straight line that is timelike and past directed. Again, we call this line $\sigma$ and proceed from here along the same lines as before in order to show that $\sigma$ (with the past endpoint $(f(\underline{x}_\infty), \underline{x}_\infty)$  deleted) corresponds to a past directed and past inextendible timelike curve in $M$ which lies completely in $I^+(\Sigma, M)$. This is again a contradiction.
Indeed, this argument, together with the uniform bounds on the metric components $\tilde{g}_{\mu \nu}$ also shows that $f$ satisfies a Lipschitz condition. 

In order to show that the graph of $f$ is achronal in $(-\varepsilon_0, \varepsilon_0) \times (-\varepsilon_1, \varepsilon_1)^d$, we first observe that we have in fact shown
$$\{(x_0,\underline{x}) \in (\varepsilon_0, \varepsilon_0) \times (-\varepsilon_1, \varepsilon_1)^{d} \; | \: x_0 < f(\underline{x})\} \subseteq \tilde{\varphi} \big( \iota\big(I^+(\Sigma,M) \big)\cap \tilde{U}\big)\;.$$
Given now two points $q,r$ in the graph of $f$ with  $r \in I^+(q, (-\varepsilon_0, \varepsilon_0)\times (-\varepsilon_1, \varepsilon_1)^d)$, we can move $r$ down a bit in the $x_0$ coordinate to obtain a point which  lies in $ I^+(q, (-\varepsilon_0, \varepsilon_0)\times (-\varepsilon_1, \varepsilon_1)^d) \cap \tilde{\varphi}\big(\iota\big(I^+(\Sigma, M)\big) \cap \tilde{U}\big)$. 
It follows that we can again find a past directed and past inextendible timelike curve in $M$ which lies completely in $I^+(\Sigma,M)$. This finishes the proof.
\end{proof}

\begin{remark}
The proof would simplify drastically if one knew that one can find a neighbourhood of $p$ in $\tilde{M}$ that is disjoint from a Cauchy hypersurface of $M$. In general this is however not the case as illustrated by the point $s$ in Figure \ref{FigPosStr}  and we need to resort to the construction illustrated in Figure \ref{FigCont} in order to ensure that the past directed timelike curve $\sigma$ remains in $I^+(\Sigma, M)$ and does not intersect the Cauchy hypersurface $\Sigma$.
\end{remark}

\section{Auxiliary results}

The following theorem was proven in \cite{GalLinSbi17}.
\begin{theorem} \label{ThmGLS}
Let $(M,g)$ be a globally hyperbolic Lorentzian manifold and $\iota : M \hookrightarrow \tilde{M}$ a $C^0$-extension. Assume that $\partial^+ \iota(M) \neq \emptyset$. Then there exists a future directed timelike geodesic $\tau : [-1,0) \to M$ that is future inextendible in $M$ and such that $\lim_{s \to 0} (\iota \circ \tau)(s)$ exists and is contained in $\partial\iota(M)$.
\end{theorem}
In fact, the proof in \cite{GalLinSbi17} together with Proposition \ref{FundProp} gives
\begin{theorem}\label{ThmGeodesic}
Let $(M,g)$ be a globally hyperbolic Lorentzian manifold and $\iota : M \hookrightarrow \tilde{M}$ a $C^0$-extension. Assume that $\partial^+ \iota(M) \neq \emptyset$ and let $p \in \partial^+ \iota(M)$. Let $\tilde{\varphi} : \tilde{U} \to (-\varepsilon_0, \varepsilon_0) \times (-\varepsilon_1, \varepsilon_1)^d$ be a chart around $p$ as in Proposition \ref{FundProp}. Then there exists a future directed timelike geodesic $\tau : [-1,0) \to M$ that is future inextendible in $M$ and such that $\tilde{\varphi} \circ \iota \circ \tau :[-1,0) \to (-\varepsilon_0, \varepsilon_0) \times (-\varepsilon_1, \varepsilon_1)^d$ maps into $\{(s,\underline{x}) \in (-\varepsilon_0, \varepsilon_0) \times (-\varepsilon_1, \varepsilon_1)^{d} \; | \: s < f(\underline{x})\}$ and has an endpoint on $\{(s,\underline{x}) \in (-\varepsilon_0, \varepsilon_0) \times (-\varepsilon_1, \varepsilon_1)^{d} \; | \: s = f(\underline{x})\}$.
\end{theorem}
In particular, Theorem \ref{ThmGeodesic} improves on Theorem \ref{ThmGLS} by guaranteeing that the future endpoint of $\iota \circ \tau$ is contained in the \emph{future} boundary of $M$.

We recall from Section 4.1 of \cite{Sbie15} the definition of the interior $(\Mint, \gint)$ of the $d+1$-dimensional Schwarzschild spacetime, where $\Mint := \R \times \big(0,(2m)^{\frac{1}{d-2}}\big) \times \mathbb{S}^{d-1}$ and
\begin{equation*}
\gint = -\big(1 - \frac{2m}{r^{d-2}}\big) \,dt^2 + \big(1 - \frac{2m}{r^{d-2}}\big)^{-1}\,dr^2 + r^2\, \mathring{\gamma}_{d-1} \;.
\end{equation*}
Here, $d \in \N_{\geq 3}$, $(t,r)$ are the standard coordinates on the first two factors of $\Mint$, $m>0$ is a parameter, and $\mathring{\gamma}_{d-1}$ denotes the standard round metric on $\mathbb{S}^{d-1}$. We define a time-orientation on $(\Mint,\gint)$ by stipulating that $-\frac{\partial}{\partial r}$ is future directed. 
For the definition of the $d+1$-dimensional maximal analytic Schwarzschild spacetime $(\Mmax, \gmax)$ we refer the reader to Section 4.1 of \cite{Sbie15}.

We need the following facts about the Schwarzschild spacetime:	

\begin{theorem}\label{ThmFact}
Let $\tau : [-1, 0) \to \Mmax$ be a future directed and future inextendible timelike geodesic (not necessarily affinely parametrised) in the maximal analytic Schwarzschild spacetime. Then we either have $(r \circ \tau)(s) \to 0$ for $s \to 0$ or $\tau$ is future complete.
\end{theorem}

For the proof of this theorem see for example Proposition 36 in Chapter 13 of \cite{ONeill}. The next two results were proven in Section 6 of \cite{Sbie15}. 

\begin{lemma}
\label{BoundsOnReach}
Let $0<r_0< (2m)^{\frac{1}{d-2}}$. For every $\varepsilon >0$ we can find $0<\tilde{r}_0 < r_0 $ such that for any future directed timelike curve $\sigma : (-r_0,0) \to \Mint$, $\sigma(s) = \big(\sigma_t(s), -s, \sigma_{\omega}(s)\big)$, where $\sigma_{\omega}$ is the canonical projection of $\sigma$ on the sphere $\Sd$, we have for all  $-\tilde{r}_0 \leq s,s' <0$
\begin{equation*}
d_{\Sd}\big(\sigma_{\omega}(s), \sigma_{\omega}(s') \big) < \varepsilon    \qquad \textnormal{ and } \qquad   |\sigma_t(s) - \sigma_t(s')| < \varepsilon \;.
\end{equation*}
\end{lemma}
Note that this lemma implies, in particular, that $\sigma_t(s)$ and $\sigma_\omega(s)$ converge for $s \nearrow 0$. We include its proof here because part of it will be referred to later on.
\begin{proof}
Let $\sigma : (-r_0,0) \to \Mint$ be a timelike curve, parametrised as above. We obtain for all $s \in (-r_0,0)$
\begin{equation*}
0 > \gint(\dot{\sigma}(s),\dot{\sigma}(s)) = -\big(1- \frac{2m}{(-s)^{d-2}}\big) (\dot{\sigma}_t(s))^2 + \big(1- \frac{2m}{(-s)^{d-2}}\big)^{-1} + s^2 \, \mathring{\gamma}_{d-1}\big(\dot{\sigma}_{\omega}(s), \dot{\sigma}_{\omega}(s)\big) 
\end{equation*}
and hence
\begin{equation*}
\frac{(-s)^{d-2}}{2m - (-s)^{d-2}}  >  \underbrace{\frac{2m - (-s)^{d-2}}{(-s)^{d-2}} (\dot{\sigma}_t(s))^2}_{\geq 0}  + \underbrace{ s^2 \, \mathring{\gamma}_{d-1}\big(\dot{\sigma}_{\omega}(s), \dot{\sigma}_{\omega}(s)\big)}_{\geq 0} \;.
\end{equation*}
It follows that
\begin{equation}
\label{BoundTVel}
|\dot{\sigma}_t(s)| < \frac{(-s)^{d-2}}{2m - (-s)^{d-2}} 
\end{equation}
and
\begin{equation*}
||\dot{\sigma}_\omega(s)||_{\Sd} < \frac{(-s)^{\nicefrac{d}{2} - 2}}{\big[2m - (-s)^{d-2}\big]^{\nicefrac{1}{2}}}
\end{equation*}
holds for all $s \in (-r_0,0)$. Since $d \geq 3$, it follows that both upper bounds are integrable on $(-r_0,0)$. The lemma now follows from integration.
\end{proof}

\begin{proposition}
\label{IntFutCon}
The interior of the Schwarzschild spacetime $(\Mint, \gint)$ is future one-connected.
\end{proposition}

\section{The proof of the $C^0$-inextendibility of the Schwarzschild spacetime}

\begin{proof}[Proof of Theorem \ref{MainThm}:]
The proof is by contradiction, so we assume that there exists a $C^0$-extension $\iota : \Mmax \hookrightarrow \tilde{M}$ of the maximal analytic Schwarzschild spacetime $(\Mmax, \gmax)$. We recall that Lemma 2.17 of \cite{Sbie15} implies that $\partial^+ \iota(\Mmax) \cup \partial^- \iota(\Mmax) \neq \emptyset$. Without loss of generality we assume that $\partial^+ \iota(\Mmax) \neq \emptyset$, otherwise we reverse the time orientation. By Theorem \ref{ThmGeodesic} there is a chart
$\tilde{\varphi} : \tilde{U} \to(-\varepsilon_0, \varepsilon_0) \times  (-\varepsilon_1, \varepsilon_1)^{d}$, $\varepsilon_0, \varepsilon_1 >0$, and a future directed timelike geodesic $\tau : [-1,0) \to \Mmax$ that is future inextendible in $\Mmax$ with the following properties:
\begin{enumerate}[i)]
\item $|\tilde{g}_{\mu \nu} - m_{\mu \nu}| < \delta_0$ (where $\delta_0>0$ is fixed below).
\item There exists a Lipschitz continuous function $f : (-\varepsilon_1, \varepsilon_1)^d \to (-\varepsilon_0, \varepsilon_0)$ with the following property: 
\begin{equation}
\{(x_0,\underline{x}) \in (-\varepsilon_0, \varepsilon_0) \times (-\varepsilon_1, \varepsilon_1)^{d} \; | \: x_0 < f(\underline{x})\} \subseteq \tilde{\varphi}\big( \iota(\Mmax)\cap \tilde{U}\big)
\end{equation} and 
\begin{equation}
\{(x_0,\underline{x}) \in (-\varepsilon_0, \varepsilon_0) \times (-\varepsilon_1, \varepsilon_1)^{d}  \; | \: x_0 = f(\underline{x})\} \subseteq \tilde{\varphi}\big(\partial^+\iota(\Mmax)\cap \tilde{U}\big) \;.
\end{equation}
Moreover, the graph of $f$ is achronal in $\tilde{U}$.
\item $\tilde{\varphi} \circ \iota \circ \tau :[-1,0) \to (-\varepsilon_0, \varepsilon_0) \times (-\varepsilon_1, \varepsilon_1)^d$ maps into $\{(x_0,\underline{x}) \in (-\varepsilon_0, \varepsilon_0) \times (-\varepsilon_1, \varepsilon_1)^{d} \; | \: x_0 < f(\underline{x})\}$ and $\lim_{s \to 0} (\tilde{\varphi} \circ \iota \circ \tau)(s) = (0, \ldots, 0)$.\footnote{This can be guaranteed after recentering the chart.}
\end{enumerate}
We will also introduce the abbreviation $R_{\varepsilon_0, \varepsilon_1} := (-\varepsilon_0, \varepsilon_0) \times ( -\varepsilon_1, \varepsilon_1)^d$.
Let $ 0 < a < 1$ and let $< \cdot, \cdot>_{R^{d+1}}$ denote the Euclidean inner product on $\R^{d+1}$ and $| \cdot |_{\R^{d+1}}$ the associated norm. We introduce the following notation: 
\begin{itemize}
\item $C^+_a := \big{\{} X \in \R^{d+1} \, | \, \frac{<X,e_0>_{\R^{d+1}}}{|X|_{\R^{d+1}}}   > a \big{\}}$
\item $C^-_a := \big{\{} X \in \R^{d+1} \, | \, \frac{<X,e_0>_{\R^{d+1}}}{|X|_{\R^{d+1}}}   < -a \big{\}}$
\item $C^c_a := \big{\{} X \in \R^{d+1} \, | \, -a < \frac{<X,e_0>_{\R^{d+1}}}{|X|_{\R^{d+1}}}   < a \big{\}}$\;.
\end{itemize}
Here, $C^+_a$ is the forward cone of vectors which form an angle of less than $\cos^{-1}(a)$ with the $x_0$-axis, and $C^-_a$ is the corresponding backwards cone. In Minkowski space, the forward and backward cones of timelike vectors correspond to the value $a = \cos(\frac{\pi}{4}) = \frac{1}{\sqrt{2}}$.

Since $\frac{5}{8} < \frac{1}{\sqrt{2}} < \frac{5}{6}$, we can now choose $\delta_0 >0$ such that in the chart $\tilde{\varphi}$ from above all vectors in $C^+_{\nicefrac{5}{6}}$ are future directed timelike, all vectors in $C^-_{\nicefrac{5}{6}}$ are past directed timelike, and all vectors in $C^c_{\nicefrac{5}{8}}$ are spacelike. It is straightforward then to prove\footnote{See also \cite{Sbie15}, Step 1.2 of the proof of Theorem 3.1.} the following estimates for $x \in \ed$:
\begin{equation}
\label{EstimatesOnPastAndFuture}
\begin{split}
&\big( x + C^+_{\nicefrac{5}{6}}\big) \cap  \ed \subseteq I^+(x, \ed) \subseteq \big( x + C^+_{\nicefrac{5}{8}}\big) \cap  \ed \\
&\big( x + C^-_{\nicefrac{5}{6}}\big) \cap  \ed \subseteq I^-(x, \ed) \subseteq \big( x + C^-_{\nicefrac{5}{8}}\big) \cap  \ed  \;.
\end{split}
\end{equation}
By Theorem \ref{ThmFact} the future directed timelike geodesic $\tau : [-1,0) \to \Mmax$ is either future complete (i.e., it `leaves through the exterior') or $(r \circ \tau)(s) \to 0$ holds for $s \to 0$ (i.e., it `leaves through the interior'). We first show that it cannot leave through the exterior.

In order to ease notation we also denote $ \tilde{\varphi} \circ \iota \circ \tau$ in the following by $\tau$.
Since $\dot{\tau}(s) \in C^+_{\nicefrac{5}{8}}$ for all $s \in [-1,0)$, we have $dx_0\big(\dot{\tau}(s)\big) >0$. It follows that we can reparametrise $\tau$ so that we can assume without loss of generality that  $\tau: [-s_0, 0) \to \ed$ is given by $\tau(s) = \big(s, \underline{\tau}(s)\big)$, where $-s_0 \in (-\varepsilon_0,0)$.
From $\dot{\tau}(s) \in C^+_{\nicefrac{5}{8}}$ for all $s \in [-s_0, 0)$, it follows that
\begin{equation*}
\frac{5}{8} < \frac{<\dot{\tau}(s), e_0>_{\R^{d+1}}}{|\dot{\tau}(s)|_{\R^{d+1}}} = \frac{1}{\sqrt{1+ |\dot{\underline{\tau}}(s)|_{\R^d}}} \;.
\end{equation*}
Hence, we obtain $|\dot{\underline{\tau}}(s)|_{\R^d} < \frac{\sqrt{39}}{5} $ for all $s \in [-s_0,0)$.
Together with the uniform bound on the metric components i), it now follows that
\begin{equation*}
\int\limits_{-s_0}^0 \sqrt{-\tilde{g}\big(\dot{\tau}(s), \dot{\tau}(s)\big)} \,ds = \int\limits_{-s_0}^0 \sqrt{ - \Big[ \tilde{g}_{00} + 2\sum_{i=1}^d \tilde{g}_{0i} \dot{\underline{\tau}}_i(s) + \sum_{i,j =1}^d \tilde{g}_{ij} \dot{\underline{\tau}}_i(s) \dot{\underline{\tau}}_j(s) \Big]}\, ds \leq C < \infty \;.
\end{equation*}
Thus, $\tau$ is not future complete. We are left with the possibility that $\tau$ leaves through the interior, i.e., $(r \circ \tau)(s) \to 0$ holds for $s \to 0$. In the following we will also derive a contradiction from this premiss.
\newline
\newline
We consider the curve $\tilde{\gamma} : (-\varepsilon_0, 0] \to \tilde{M}$ which is given in the chart $\tilde{\varphi}$ by $(-\varepsilon_0, 0] \ni s \mapsto (s, 0, \ldots, 0)$ and we set $\gamma :=  \iota^{-1} \circ \tilde{\gamma}|_{(-\varepsilon_0, 0)}$, which is a future directed and future inextendible timelike curve in $\Mmax$. We claim that $(r \circ \gamma)(s) \to 0$ for $s \to 0$. To see this, we observe that by Proposition 2.6 of \cite{Sbie15}  $I^+\big(\tilde{\gamma}(s), \tilde{U}\big)$ is an open neighbourhood of $(0, \ldots, 0)$ for every $s \in (-\varepsilon_0, 0)$. Thus, there exists a future directed timelike curve from $\tilde{\gamma}(s)$  to $(\tilde{\varphi} \circ \iota \circ \tau)\big([-1,0)\big)$ for all $s \in (-\varepsilon_0, 0)$. By the first point of Remark \ref{RemBoundary} these timelike curves lie completely in $\iota(\Mmax)$. It is easy now to conclude that $\gamma$ must also leave through the interior\footnote{Indeed, through the same point of the conformal boundary as given by the Penrose conformal compactification.}. From now on we can forget about $\tau$ -- we will only work with $\gamma$ now. The rest of the proof proceeds in three steps.

\underline{\textbf{Step 1:}} We show that there exists a $\mu >0$ such that
\begin{enumerate}
\item $\iota \Big( I^+\big( \gamma(-\mu), \Mint\big)\Big) \subseteq \tilde{\varphi}^{-1} \Big(\{(x_0,\underline{x}) \in (\varepsilon_0, \varepsilon_0) \times (-\varepsilon_1, \varepsilon_1)^{d} \; | \: x_0 < f(\underline{x})\}\Big)$
\item $(-\varepsilon_0, -\frac{49}{50} \varepsilon_0) \times (-\varepsilon_1, \varepsilon_1)^d \subseteq I^-\big((\tilde{\varphi}\circ \tilde{\gamma})(-\mu), (-\varepsilon_0, \varepsilon_0) \times (-\varepsilon_1, \varepsilon_1)^d\big) \;.$
\end{enumerate}
\vspace*{2mm}

\textbf{Step 1.1}
As in the proof of Proposition \ref{FundProp} we can without loss of generality assume that $0< \varepsilon_1 < \frac{1}{2} \varepsilon_0$. We now choose $x^+ :=(x^+_0, 0, \ldots, 0)$, $0< x^+_0 <\varepsilon_0$, and $x^- :=(x^-_0, 0, \ldots, 0)$, $-\varepsilon_0 < x^-_0 <0$ such  that the closure of $\big(x^+ + C^-_{\nicefrac{5}{6}}\big) \cap \big(x^- + C^+_{\nicefrac{5}{6}}\big)$ in $\ed$ is compact.
Now choose $y^- := (y^-_0, 0, \ldots, 0)$ with $ - \frac{1}{5}\varepsilon_0 < y^-_0 < 0$ so that the closure of $C^-_{\nicefrac{5}{8}} \cap \big(y^- + C^+_{\nicefrac{5}{8}} \big)$ in $\ed$ is contained in  $x^+ + C^-_{\nicefrac{6}{7}} \cap x^- + C^+_{\nicefrac{6}{7}}$. Moreover, since we have $- \frac{1}{5}\varepsilon_0 < y^-_0$ and $0< \varepsilon_1 < \frac{1}{2} \varepsilon_0$ it follows that 
\begin{equation}\label{EqBottomChart}
(-\varepsilon_0, -\frac{49}{50} \varepsilon_0) \times (-\varepsilon_1, \varepsilon_1)^d \subseteq I^-\big(y^-, \ed\big) \;.
\end{equation} 
\begin{figure}[h]
  \centering
\def\svgwidth{7.5cm}
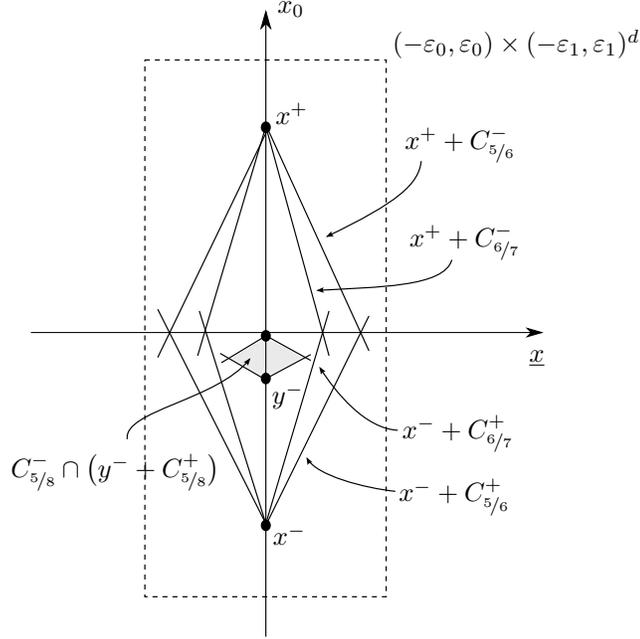
  \caption{The set-up of Step 1.1} \label{FigSetUp}
\end{figure}
We claim that for all $y^-_0 < s < 0$ we have 
\begin{equation}
\label{FuturePastPreserving}
\tilde{\varphi}^{-1} \Big(I^-\big((s,0, \ldots, 0), \ed\big) \cap I^+\big(y^-, \ed\big) \Big) \\
= \iota \Big(I^-\big(\gamma(s), \Mint \big) \cap I^+\big( \gamma(y^-_0), \Mint\big) \Big) \;.
\end{equation}

The inclusion ``$\, \subseteq \,$'' follows from the first point of Remark \ref{RemBoundary}, since if $\sigma$ is a past directed timelike curve in $\ed$ from $(s,0, \ldots, 0)$ to $y^-$, then the first point of Remark \ref{RemBoundary} states that $\tilde{\varphi}^{-1} \circ \sigma$ is contained in $\iota(\Mint)$.

To prove ``$\, \supseteq \,$'', let $\sigma : [y_0^-, s] \to \Mint$ be a future directed timelike curve from $\gamma(y_0^-)$ to $\gamma(s)$. By Proposition \ref{IntFutCon}, there exists a timelike homotopy $\Gamma :  [0,1]\times [y^-_0,s] \to \Mint$ with fixed endpoints between $\gamma |_{[y^-_0,s]}$ and $\sigma$. It follows that $\iota \circ \Gamma: [0,1] \times [y^-_0,s]  \to \tilde{M}$ is a timelike homotopy with fixed endpoints in $\tilde{M}$. We need to show that $(\iota \circ \sigma) (\cdot)= (\iota \circ \Gamma)(1, \cdot)$ maps into $\tilde{U}$. We argue by continuity, i.e., we show that the interval $J := \{t \in [0,1] \,|\, \iota \circ \Gamma\big(t, [y^-_0,s]\big) \subseteq \tilde{U}\}$ is non-empty, open and closed in $[0,1]$.

Clearly, we have $0 \in J$, since $\iota \circ \Gamma(0, \cdot) = \tilde{\gamma}|_{[y^-_0,s]}$. The openness follows from the openness of $\tilde{U}$, and the closedness follows since $I^-\big(\tilde{\gamma}(s), \tilde{U}\big) \cap I^+ \big(\tilde{\gamma}(y^-_0),\tilde{U}\big)$ is precompact in $\tilde{U}$, i.e., in particular its closure in $\tilde{M}$ is contained in $\tilde{U}$.
This finishes the proof of \eqref{FuturePastPreserving}.

Moreover, it is easy to see\footnote{See also \cite{Sbie15}, Proposition 2.7} that the following holds
\begin{equation*}
I^-\big(0, \ed\big) \cap I^+\big(y^-, \ed\big) = \bigcup_{y_0^- < s < 0} \Big(I^-\big((s,0, \ldots, 0), \ed\big) \Big)\cap I^+\big(y^-, \ed\big) \;.
\end{equation*} 
Together with \eqref{FuturePastPreserving} this now implies
\begin{equation}
\label{FuturePastPreserving2}
\tilde{\varphi}^{-1} \Big(I^-\big(0, \ed\big) \cap I^+\big(y^-, \ed\big) \Big) \\
=  \iota \Big(\Big(\bigcup_{y_0^- < s < 0} I^-\big(\gamma(s), \Mint \big)\Big) \cap I^+\big( \gamma(y^-_0), \Mint\big) \Big) \;.
\end{equation}
\vspace*{2mm}

For the next step we recall from \cite{Sbie15}, Definition 2.14, that for a Lorentzian manifold $(M,g)$ two sets $A,B \subseteq M$ are called \emph{timelike separated by a set $K \subseteq M$}, iff every timelike curve connecting $A$ and $B$ intersects $K$.
\vspace*{2mm}

\textbf{Step 1.2} In this step we switch back to the manifold $(\Mint,\gint)$. Choose a $y^+_0$ with $y^-_0 < y^+_0 <0$ and define 
\begin{equation}
\label{DefK}
K:=\Big[\Big(\bigcup_{y_0^- < s < 0} I^-\big(\gamma(s), \Mint \big)\Big) \cap I^+\big( \gamma(y^-_0), \Mint\big) \Big] \setminus I^+\big( \gamma(y^+_0), \Mint\big)  \;.
\end{equation}
\vspace*{2mm}

\textbf{Step 1.2.1}  We show that the set $K$ timelike separates the set $\gamma\big((y^+_0,0)\big)$ from $I^-\big(\gamma(y^-_0),\Mint\big)$.
\vspace*{2mm}

So let $\sigma : [0,1] \to \Mint$ be a past directed timelike curve with $\sigma (0) \in \gamma\big((y^+_0,0)\big)$ and $\sigma(1) \in I^-\big(\gamma(y^-_0),\Mint\big)$. 

\begin{figure}[h]
\centering
\def\svgwidth{4.8cm}
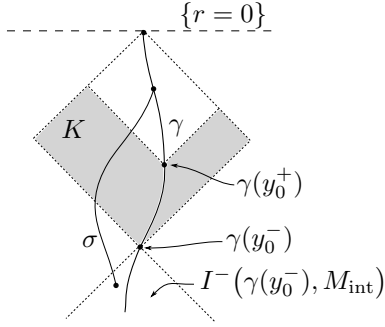
\caption{For the proof of Step 1.2.1}
\end{figure}

We claim that there exists a $\Delta_- \in (0,1)$ such that 
\begin{equation*}
\sigma^{-1}\big[I^+\big(\gamma(y^-_0),\Mint\big)\big] = [0,\Delta_-)\;.
\end{equation*} 
This is seen as follows: To begin with, it is clear that $0 \in \sigma^{-1}\big[I^+\big(\gamma(y^-_0),\Mint\big)\big]$. Moreover, by the continuity of $\sigma$ and the openness of $I^+\big(\gamma(y^-_0),\Mint\big)$, we know that $\sigma^{-1}\big[I^+\big(\gamma(y^-_0),\Mint\big)\big]$ is open in $[0,1]$. Moreover, since $\sigma$ is a past directed timelike curve, it follows that if $s_0 \in \sigma^{-1}\big[I^+\big(\gamma(y^-_0),\Mint\big)\big]$, then we also have $[0,s_0] \in \sigma^{-1}\big[I^+\big(\gamma(y^-_0),\Mint\big)\big]$. And finally, since $(\Mint,\gint)$ satisfies the chronology condition, there are no closed timelike curves in $\Mint$, and hence $I^-\big(\gamma(y^-_0),\Mint\big)$ is disjoint from $I^+\big(\gamma(y^-_0),\Mint\big)$. This implies that $\Delta_- <1$.

In the same way we deduce that there exists a $\Delta_+ \in (0,1)$ such that
\begin{equation*}
\sigma^{-1}\big[I^+\big(\gamma(y^+_0),\Mint\big)\big] = [0,\Delta_+)\;.
\end{equation*} 
In the following we show that $\Delta_+ < \Delta_-$.

Since $(\Mint,\gint)$ is globally hyperbolic and $\gint$ is smooth (!), we have $\overline{I^+\big(\gamma(y^+_0),\Mint\big)} = J^+\big(\gamma(y^+_0),\Mint\big)$.\footnote{This follows from 6.\ Lemma and 22.\ Lemma of Chapter 14 of \cite{ONeill}. Here $J^+\big(\gamma(y^+_0),\Mint\big)$ denotes the causal future of $\gamma(y^+_0)$ in $\Mint$, see also Chapter 14 of \cite{ONeill}.} Together with $\gamma(y^+_0) \in I^+\big(\gamma(y^-_0),\Mint \big)$, we now obtain\footnote{See 1.\ Corollary in Chapter 14 of \cite{ONeill}.}
\begin{equation*}
\overline{I^+\big(\gamma(y^+_0),\Mint\big)} = J^+\big(\gamma(y^+_0),\Mint\big) \subseteq J^+\Big(I^+\big(\gamma(y^-_0),\Mint\big),\Mint\Big) = I^+\big(\gamma(y^-_0),\Mint\big) \;.
\end{equation*}
Hence, we have $\sigma(\Delta_+) \in I^+\big(\gamma(y^-_0),\Mint\big)$, from which it follows that $\Delta_+ < \Delta_-$.

Choosing $s_0 \in (\Delta_+,\Delta_-)$, it follows that $\sigma(s_0) \in   I^+\big(\gamma(y^-_0),\Mint\big)  \setminus I^+\big(\gamma(y^+_0),\Mint\big)$. Moreover, it is clear that $\sigma(s_0) \in \bigcup_{-\varepsilon < s < 0} I^-\big(\gamma(s), \Mint \big)$, which concludes Step 1.2.1.
\vspace*{2mm}

\textbf{Step 1.2.2} We show that $\overline{K}$ is compact.
\vspace*{2mm}

We claim that for any $s_0 \in (y^+_0,0)$ there exists a $\delta >0$ and a neighbourhood $V$ of $\mathrm{id} \in SO(d)$ such that
\begin{equation*}
\big(\gamma_t(s) - \delta, \gamma_t(s) + \delta\big) \times \big{\{}\gamma_r(s)\big{\}} \times \big{\{}f \cdot \gamma_\omega(s) \, | \, f \in V\big{\}} \subseteq I^+\big(\gamma(y^+_0),\Mint\big)
\end{equation*}
holds for all $s \in (s_0,0)$.

In order to prove this claim, we first note that since $\gamma$ is timelike and $\gint$ is continuous, there exists a $\mu >0$ such that 
\begin{equation}
\label{UniformTimelike}
\gint\big(\dot{\gamma}(s), \dot{\gamma}(s)\big) < - \mu
\end{equation}
holds for all $s \in [y^+_0,s_0]$.

Let now $\lambda \in C^\infty\big( [y^+_0,s_0], \R\big)$ and $ h \in C^\infty\big([y^+_0,s_0], SO(d) \subseteq \mathrm{Mat}(d \times d, \R)\big)$, and define $\sigma : [y^+_0,s_0] \to \Mint$ by
\begin{equation*}
\sigma(s) := \Big(\gamma_t(s) + \lambda(s), \gamma_r(s), h(s)\big(\gamma_\omega(s)\big)\Big) \;.
\end{equation*}
We compute
\begin{equation}
\label{NormSigma}
\begin{split}
\gint\big(\dot{\sigma},\dot{\sigma}\big) = - \Big(1 - \frac{2m}{\big(\gamma_r(s)\big)^{d-2}}\Big)\, \big(\dot{\gamma}_t(s) + \dot{\lambda}(s)\big)^2 &+ \Big(1 - \frac{2m}{\big(\gamma_r(s)\big)^{d-2}}\Big)^{-1} \, \big(\dot{\gamma}_r(s)\big)^2 \\
&+ \big(\gamma_r(s)\big)^2 \, \big|\big|\dot{h}(s) \gamma_\omega(s) + h(s) \dot{\gamma}_\omega(s)\big|\big|^2_{\R^d} \;,
\end{split}
\end{equation}
where $|| \cdot ||_{\R^d}$ denotes the Euclidean norm on $\R^d$, and we think of $\gamma_\omega$ as mapping into $\Sd \subseteq \R^d$.
Since $\gamma_r(s)$ is bounded away from $0$ for $s \in [y^+_0, s_0]$, we can infer from \eqref{UniformTimelike} and \eqref{NormSigma} that there exists an $\eta >0$ such that whenever 
\begin{equation}
\label{BoundsOnSlopes}
||\dot{\lambda}||_{L^\infty\big([y^+_0,s_0]\big)} + ||\dot{h}||_{L^\infty\big([y^+_0,s_0]\big)} < \eta
\end{equation}
holds, the curve $\sigma : [y^+_0,s_0] \to \Mint$ is timelike\footnote{One can for example define $||\dot{h}||_{L^\infty\big([y^+_0,s_0]\big)} := \sup\limits_{s \in [y^+_0,s_0]} ||\dot{h}(s)||_{\R^{d\times d}}$.}. This in turn implies the existence of a $\delta >0$ and a neighbourhood $V$ of $\mathrm{id} \in SO(d)$ such that for every $\lambda_{s_0} \in (-\delta, \delta)$ and for every $h_{s_0} \in V$ there are smooth functions 
\begin{equation*}
\lambda \in C^\infty\big( [y^+_0,s_0], \R\big) \quad \textnormal{ with } \; \lambda(y^+_0) = 0 \; \textnormal{ and } \; \lambda(s_0) = \lambda_{s_0}
\end{equation*}
and
\begin{equation*}
h \in C^\infty\big([y^+_0,s_0], SO(d) \subseteq \mathrm{Mat}(d \times d, \R)\big) \quad \textnormal{ with }\; h(y^+_0) = \mathrm{id} \; \textnormal{ and } \; h(s_0) = h_{s_0}
\end{equation*}
such that moreover \eqref{BoundsOnSlopes} is satisfied. The claim now follows from concatenating $\sigma$ with the timelike\footnote{Recall that the Schwarzschild metric $\gint = - ( 1- \frac{2m}{r^{d-2}}) \, dt^2 + ( 1 - \frac{2m}{r^{d-2}})^{-1} \, dr^2 + r^2 \mathring{\gamma}_{d-1}$ is spherically symmetric and invariant under translations in $t$.} curve  $\tau : [s_0, 0) \to \Mint$ given by
\begin{equation*}
\tau(s) = \Big(\gamma_t(s) + \lambda_{s_0}, \gamma_r(s), h_{s_0} \big(\gamma_\omega(s)\big)\Big) \;.
\end{equation*}

We now fix $s_0 \in (y^+_0, 0)$ and obtain $\delta >0$ and a neighbourhood $V \subseteq SO(d)$ of $\mathrm{id} \in SO(d)$ as in the claim. It follows from Lemma \ref{BoundsOnReach} that there exists $s_1 \in (s_0 ,0)$ (close to $0$) such that
\begin{equation*}
\begin{split}
I^-\big(\gamma(s),\Mint\big) \cap \Big{\{} r \leq \gamma_r(s_1)\Big{\}} &\subseteq \bigcup_{s_0  \leq s' <0} \Big[ \big(\gamma_t(s') - \delta, \gamma_t(s') + \delta\big) \times \big{\{}\gamma_r(s')\big{\}} \times \big{\{}f \cdot \gamma_\omega(s') \, | \, f \in V\big{\}} \Big] \\
&\subseteq I^+\big(\gamma(y^+_0),\Mint\big)
\end{split}
\end{equation*}
holds for all $s \in (-\varepsilon_0, 0)$. This implies 
\begin{equation*}
K \subseteq \R \times \big(\gamma_r(s_1), \gamma_r(y^-_0)\big) \times \Sd \;.
\end{equation*}
Moreover, the bound \eqref{BoundTVel} implies that there are $t_0, t_1 \in \R$ such that 
\begin{equation*}
I^+\big(\gamma(y^-_0),\Mint\big) \subseteq (t_0,t_1) \times \big(0, \gamma_r(y^-_0)\big) \times \Sd \;.
\end{equation*}
It follows that 
\begin{equation*}
K \subseteq (t_0,t_1) \times \big(\gamma_r(s_1), \gamma_r(y^-_0)\big) \times \Sd\;,
\end{equation*}
which implies that $\overline{K}$ is compact.
\vspace*{2mm}

\textbf{Step 1.3} First note that it follows from the definition of $K$, \eqref{DefK}, that $K \subseteq \Big(\bigcup_{-\varepsilon_0 < s < 0} I^-\big(\gamma(s), \Mint \big)\Big) \cap I^+\big( \gamma(y^-_0), \Mint\big)$, and thus, together with \eqref{FuturePastPreserving2}, we obtain
\begin{equation*}
\overline{\tilde{\varphi} \big( \iota(K)\big)}  \subseteq \overline{I^-\big(0, \ed\big) \cap I^+\big(y^-, \ed\big)} \;.
\end{equation*}
By the choice of $y^- \in \ed$ in Step 1.1, together with \eqref{EstimatesOnPastAndFuture}, it follows that
\begin{equation}
\label{Neighbourhood}
\overline{\tilde{\varphi} \big( \iota(K)\big)}  \subseteq  \Big(x^+ + C^-_{\nicefrac{6}{7}}\Big) \cap \Big(x^- + C^+_{\nicefrac{6}{7}}\Big)\;.
\end{equation}
Moreover, the continuity of $\tilde{\varphi}$ and $\iota$ implies\footnote{Indeed, since $\overline{K} \subseteq \Mint$ is compact, we actually have equality.}
\begin{equation}
\label{ClosureRel}
\tilde{\varphi}\big(\iota(\overline{K}) \big) \subseteq \overline{\tilde{\varphi} \big( \iota(K)\big)} \;.
\end{equation}
Hence, from \eqref{ClosureRel} and \eqref{Neighbourhood}, it follows that 
\begin{equation*}
W:= (\tilde{\varphi} \circ \iota)^{-1}\Big(\big[x^+ + C^-_{\nicefrac{6}{7}}\big] \cap \big[x^- + C^+_{\nicefrac{6}{7}}\big]\Big) \subseteq \Mint
\end{equation*}
is an open neighbourhood of $\overline{K} \subseteq \Mint$.
\vspace*{2mm}

\textbf{Step 1.4} We show that there exists a $\mu >0$ such that $I^+\big(\gamma(-\mu), \Mint\big)$ is timelike separated from $\gamma(x^-_0)$ by $W$.
\vspace*{2mm}

We consider $\Mint = \R \times \big(0,(2m)^{\frac{1}{d-2}}\big) \times \Sd$ with the metric $d_{\Mint} : \Mint \times \Mint \to [0,\infty)$ given by
\begin{equation*}
d_{\Mint}\Big((t_1, r_1, \omega_1), (t_2, r_2, \omega_2)\Big) := |t_1 - t_2| + |r_1 - r_2| + d_{\Sd}\big(\omega_1, \omega_2\big) \;,
\end{equation*}
where $(t_i, r_i, \omega_i) \in \Mint$ for $i=1,2$. Since  
\begin{equation*}
\Mint \ni (t,r,\omega) \mapsto d_{\Mint}\big((t,r,\omega),\Mint \setminus W\big) = \inf\limits_{(t',r',\omega') \in \Mint \setminus W} d_{\Mint}\big((t,r,\omega),(t',r',\omega')\big)
\end{equation*}
is continuous, and $\overline{K}$ is compact and disjoint from the closed set $\Mint \setminus W$, we infer that $d_{\Mint}(\cdot, \Mint \setminus W)$ must attain its minimum on $\overline{K}$, which is moreover strictly positive. It follows that there exists a $\delta >0$ such that 
\begin{equation*}
\overline{K}_\delta := \Big{\{} (t,r,\omega) \in \Mint \, \big| \, d_{\Mint}\big((t,r,\omega), \overline{K}\big) < \delta \Big{\}} \subseteq W \;.
\end{equation*}
Moreover, by choosing $\delta $ slightly smaller if necessary, we can also assume that 
\begin{equation}
\label{RestrictionOnDelta}
B_\delta\big(\gamma(x^-_0)\big) \subseteq I^-\big(\gamma(y^-_0),\Mint\big)\;.
\end{equation}
\vspace*{2mm}

\textbf{Step 1.4.1}
We define a metric $d_{SO(d)} : SO(d) \times SO(d) \to [0,\infty)$ on $SO(d)$ by
\begin{equation*}
d_{SO(d)}(f,h) := \sup\limits_{\omega \in \Sd} d_{\Sd}\big(f (\omega), h(\omega)\big)\;,
\end{equation*}
where $f, h  \in SO(d)$, and denote with $B_{\eta}(\mathrm{id}) \subseteq SO(d)$ the ball of radius $\eta >0 $, centred at $\mathrm{id}$, with respect to this metric. Moreover, it is easy to see that 
\begin{equation}
\label{ChooseRotation}
\textnormal{for } \omega_0, \omega_1 \in \Sd \textnormal{ with } d_{\Sd}(\omega_0, \omega_1) < \eta \textnormal{,  there exists an } h \in B_\eta(\mathrm{id})  \textnormal{ with } h(\omega_0) = \omega_1\;. 
\end{equation}
In particular, $h$ can be defined as a rotation purely in the plane $\mathrm{span}\{\omega_0, \omega_1\} \subseteq \R^d$.
\vspace*{2mm}

Continuing the proof of Step 1.4, Lemma \ref{BoundsOnReach} implies that there exists a $\mu \in (0, -y^+_0)$ such that for all $(t_0,r_0, \omega_0) \in I^+\big(\gamma(-\mu),\Mint\big)$, we have
\begin{equation}
\label{ChoiceMu}
|t_0 - \gamma_t(-\mu)| < \frac{\delta}{2} \qquad \textnormal{ and } \qquad d_{\Sd}\big(\omega_0, \gamma_\omega(-\mu)\big) < \frac{\delta}{2} \;.
\end{equation}
In the following we will show that $I^+\big(\gamma(-\mu),\Mint\big)$ is timelike separated from $\gamma(x^-_0)$ even by $\overline{K}_\delta$ (which, of course, implies Step 1.4).

So let $\sigma : [0,1] \to \Mint$ be a past directed timelike curve with $\sigma(0) \in I^+\big(\gamma(-\mu),\Mint\big)$ and $\sigma (1) = \gamma(x^-_0)$. Also let $s_0 \in (-\mu ,0)$ be such that $\gamma_r(s_0) = \sigma_r (0)$.
 By \eqref{ChoiceMu} we have 
 \begin{equation*}
 |\sigma_t(0) - \gamma_t(s_0)| < \delta \qquad \textnormal{ and } \qquad d_{\Sd}\big(\sigma_\omega(0), \gamma_\omega(s_0)\big) < \delta \;.
 \end{equation*}
Thus, by \eqref{ChooseRotation} there exists an $h \in B_\delta(\mathrm{id}) \subseteq SO(d)$ such that $h\big(\sigma_\omega(0)\big) = \gamma_\omega(s_0)$. It now follows that the curve $\hat{\sigma} : [0,1] \to \Mint$, given by 
\begin{equation*}
\hat{\sigma}(s) = \Big(\sigma_t(s) + [\gamma_t(s_0) - \sigma_t(0)], \sigma_r(s), h\big(\sigma_\omega(s)\big)\Big)\;,
\end{equation*}
is past directed timelike with $\hat{\sigma}(0) = \gamma(s_0)$ and, using the fact that $h \in B_\delta(\mathrm{id})$ together with \eqref{RestrictionOnDelta},  $\hat{\sigma}(1) \in I^-\big(\gamma(y^-_0),\Mint\big)$. By Step 1.2.1, there exists an $\hat{s} \in [0,1]$ with $\hat{\sigma}(\hat{s}) \in K$. It now follows that $\sigma(\hat{s}) \in \overline{K}_\delta$, which concludes Step 1.4.
\vspace*{2mm}

We now finish the proof of Step 1. We first show that $ \iota \big(I^+(\gamma(-\mu), \Mint)\big) \subseteq  \tilde{\varphi}^{-1}\Big(\big[x^+ + C^-_{\nicefrac{6}{7}}\big] \cap \big[x^- + C^+_{\nicefrac{6}{7}}\big]\Big)$.  

The proof is by contradiction. So let $\sigma : [0,1] \to \Mint$ be a future directed timelike curve with $\sigma(0) = \gamma(-\mu)$ and assume that there exists $\tilde{s} \in [0,1]$ such that $(\tilde{\varphi} \circ \iota \circ \sigma)(\tilde{s}) \notin \big[x^+ + C^-_{\nicefrac{6}{7}}\big] \cap \big[x^- + C^+_{\nicefrac{6}{7}}\big]$. Let
\begin{equation*}
s_0 := \sup \big{\{}s' \in [0,1] \, | \, (\tilde{\varphi} \circ \iota \circ \sigma)(s) \in \big[x^+ + C^-_{\nicefrac{6}{7}}\big] \cap \big[x^- + C^+_{\nicefrac{6}{7}}\big] \textnormal{ for all } s \in [0,s')\big{\}} \;.
\end{equation*}
Clearly, we have $0 < s_0 \leq 1$ and from our assumption it follows that $(\tilde{\varphi} \circ \iota \circ \sigma)(s_0)  \in \partial \big(\big[x^+ + C^-_{\nicefrac{6}{7}}\big] \cap \big[x^- + C^+_{\nicefrac{6}{7}}\big]\big)$. 
Since all vectors in $C^-_{\nicefrac{5}{6}}$ are past directed timelike, we can find a past directed timelike curve $\tau : [0,1] \to \ed$ with $\tau(0) = (\tilde{\varphi} \circ \iota \circ \sigma) (s_0)$ and $\tau(1) = x^-$, which does not intersect $\big[x^+ + C^-_{\nicefrac{6}{7}}\big] \cap \big[x^- + C^+_{\nicefrac{6}{7}}\big]$. For example, this curve can be chosen to lie in $\partial \Big(\big[x^+ + C^-_{\nicefrac{6}{7}}\big] \cap \big[x^- + C^+_{\nicefrac{6}{7}}\big]\Big)$, see Figure \ref{FigCon}.
\begin{figure}[h]
\centering
\def\svgwidth{9cm}
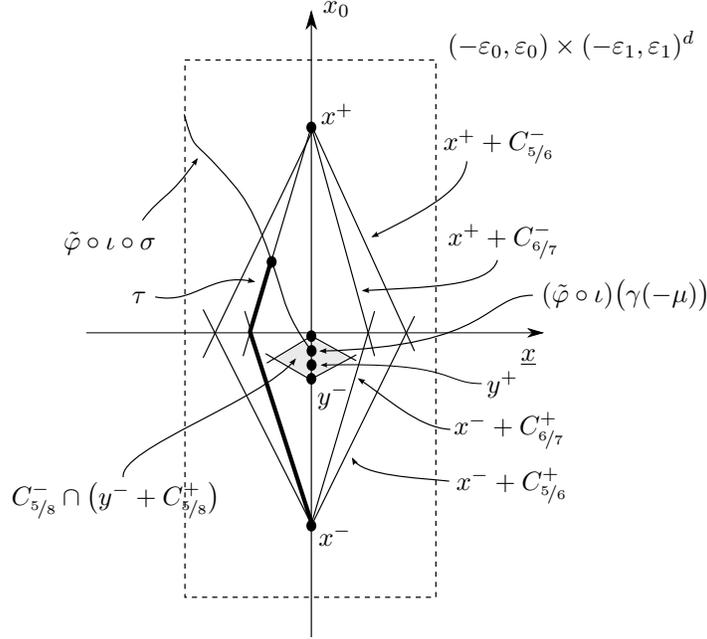
\caption{The proof of Step 1 -- by contradiction} \label{FigCon}
\end{figure}
The first point of Remark \ref{RemBoundary} implies that $\tau$ maps in $\tilde{\varphi}\big(\iota(\Mint) \cap \tilde{U}\big)$. Hence, $(\tilde{\varphi} \circ \iota)^{-1} \circ \tau$ is a past directed timelike curve in $\Mint$ with $\big((\tilde{\varphi} \circ \iota)^{-1} \circ \tau\big)(0) = \sigma(s_0) \in I^+\big(\gamma(-\mu),\Mint\big)$ and $\big((\tilde{\varphi} \circ \iota)^{-1} \circ \tau\big)(1) =\gamma(x^-_0)$, which does not intersect $W= (\tilde{\varphi} \circ \iota)^{-1}\Big(\big[x^+ + C^-_{\nicefrac{6}{7}}\big] \cap \big[x^- + C^+_{\nicefrac{6}{7}}\big]\Big)$. This, however, is a contradiction to Step 1.4. Hence, we have shown that $ \iota \big(I^+(\gamma(-\mu), \Mint)\big) \subseteq  \tilde{\varphi}^{-1} \Big(\big[x^+ + C^-_{\nicefrac{6}{7}}\big] \cap \big[x^- + C^+_{\nicefrac{6}{7}}\big]\Big)$, which in particular implies the first point of Step 1. The second point follows from \eqref{EqBottomChart} together with $y_0^- < y_0^+ < -\mu < 0$. This concludes the proof of Step 1.
\vspace*{2mm}

\underline{\textbf{Step 2:}} We show that Step 1 implies the following: There exists a constant $0<C_d < \infty$ such that for any Cauchy hypersurface $\Sigma$ of $\Mint$ the \underline{distance in $\Sigma$} of any two points in $I^+\big(\gamma(-\mu), \Mint\big) \cap \Sigma$ is bounded by $C_d$.

Explicitly, this means that for all $p,q \in I^+\big(\gamma(-\mu), \Mint\big) \cap \Sigma$ we have 
\begin{equation*}
d_{\Sigma}(p,q) := \inf_{\substack{\gamma : [0,1] \to \Sigma \textnormal{ piecewise smooth} \\\gamma(0) = p \textnormal{ and } \gamma(1) = q }}   \Big{\{} \int_0^1 \sqrt{ \overline{g}\big(\dot{\gamma}(s), \dot{\gamma}(s)\big)} \, ds\Big{\}} \leq C_d\;,
\end{equation*}
where $\overline{g}$ is the induced metric on $\Sigma$.
\vspace*{5mm}

Let $\Sigma$ be a Cauchy hypersurface of $\Mint$ and assume that $\gamma(-\mu) \in I^-(\Sigma, \Mint)$ -- otherwise there is nothing to show. For all $\underline{x} \in (-\varepsilon_1, \varepsilon_1)^d$ the future directed timelike curves 
\begin{equation*}
\sigma_{\underline{x}} : \big(-\varepsilon_0, f(\underline{x})\big) \to \Mint\;, \qquad \qquad \sigma_{\underline{x}}(s) = (\iota^{-1} \circ \tilde{\varphi}^{-1})(s,\underline{x})
\end{equation*}
are future inextendible in $\Mint$. Moreover, by the second point of Step 1 they start in $I^-\big(\gamma(-\mu), \Mint\big) \subseteq I^-(\Sigma, \Mint)$. It now follows from $\Sigma$ being a Cauchy hypersurface of $\Mint$ that for each $\underline{x} \in (-\varepsilon_1, \varepsilon_1)^d$ the curve  $\sigma_{\underline{x}}$ intersect $\Sigma$ exactly once - say at $s = h(\underline{x})$. This defines a function $h :  (-\varepsilon_1, \varepsilon_1)^d \to (-\varepsilon_0, \varepsilon_0)$.
\vspace*{2mm}

\textbf{Step 2.1:} We show that $h : (-\varepsilon_1, \varepsilon_1)^d \to (-\varepsilon_0, \varepsilon_0)$ is smooth.
\vspace*{2mm}

Let $\ux_0 \in (-\varepsilon_1, \varepsilon_1)^d$. Since $\tilde{\varphi}\big(\iota(\Sigma) \cap \tilde{U}\big)$ is a smooth submanifold of $\ed$, there exists a smooth submersion $u : W \to \R$, where $W \subseteq \ed$ is an open neighbourhood of $\big(h(\ux_0), \ux_0\big)$, such that $\tilde{\varphi}\big(\iota(\Sigma) \cap \tilde{U}\big) \cap W = \{ u=0\}$. Since $\Sigma$ is a Cauchy hypersurface, the timelike vector field $\partial_0$ can be nowhere tangent to $\tilde{\varphi}\big(\iota(\Sigma) \cap \tilde{U}\big)$. It thus follows that $\partial_0 u|_{\big(h(\ux_0), \ux_0\big)} \neq 0$. By the implicit function theorem, there is now a smooth function $v : V \to \R$, where $V \subseteq (-\varepsilon_1, \varepsilon_1)^d$ is an open neighbourhood of $\ux_0$, such that $u\big(v(\ux),\ux\big) =0$. Thus, we must have $h|_V = v$ -- and hence, $h$ is smooth.
\vspace*{2mm}

\textbf{Step 2.2:} We show that there exists a $0< C_\mathrm{slope} < \infty$ such that $|\partial_i h(\ux)| \leq C_\mathrm{slope} $ holds for all $\ux \in (-\varepsilon_1, \varepsilon_1)^d$  and for all $i = 1, \ldots, d$.
\vspace*{2mm}

Since the tangent space of a Cauchy hypersurface does not contain timelike vectors, we obtain the inequality
\begin{equation}
\label{Slope}
0 \leq \tilde{g}\Big((\partial_i h) \partial_0 + \partial_i, (\partial_i h) \partial_0 + \partial_i\Big) = (\partial_i h)^2 \tilde{g}_{00} + 2 (\partial_i h) \tilde{g}_{0i} + \tilde{g}_{ii} 
\end{equation}
for all $i = 1, \ldots, d$.
Equality in \eqref{Slope} holds for  
\begin{equation*}
(\partial_i h)_\pm =  \frac{-\,\tilde{g}_{0i} \mp \sqrt{(\tilde{g}_{0i})^2 - \tilde{g}_{ii} \tilde{g}_{00}}}{\tilde{g}_{00}} \;.
\end{equation*}
Recall that $\tilde{g}_{00} < -1 + \delta_0 < 0$ and $|\tilde{g}_{\mu \nu}| \leq 1 + \delta_0$. Hence, we obtain the uniform bound
\begin{equation*}
\max\Big{\{}\big{|}(\partial_i h)_+\big{|} \, , \,\big{|} (\partial_i h)_-\big{|} \Big{\}} \leq C_\mathrm{slope}  \;,
\end{equation*}
where $0< C_\mathrm{slope}  < \infty$ is a constant depending on $\delta_0$ (but not on $h$).
Moreover, together with $\tilde{g}_{00} <0$, the inequality \eqref{Slope} implies
\begin{equation*}
(\partial_i h)_- \leq (\partial_i h) \leq (\partial_i h)_+ 
\end{equation*}
and thus $|\partial_i h| \leq C_\mathrm{slope} $ for all $i = 1, \ldots, d$.
\vspace*{2mm}

We now define $\omega : (-\varepsilon_1, \varepsilon_1)^d \to \ed$ by $\omega(\underline{x}) = \big(h(\underline{x}), \underline{x}\big)$. Clearly, $\omega$ parametrises a smooth submanifold $\tilde{S}$ of $\ed$ which, via $\iota^{-1} \circ \tilde{\varphi}^{-1}$, is isometric to an open subset $S$ of $\Sigma \subseteq \Mint$. The ambient metric $\tilde{g}$ on $\ed$ induces a metric $\overline{g}$ on $\tilde{S}$.
Note that this metric is not necessarily positive definite. Its components with respect to the chart $\omega^{-1} $ are denoted by $\overline{g}_{ij}$, where $i,j \in \{1, \ldots, d\}$.
\vspace*{2mm}

\textbf{Step 2.3:} We show that there exists a constant $0 < C_{\overline{g}} < \infty$ such that the uniform bound $|\overline{g}_{ij}(\ux)| \leq C_{\overline{g}}$ holds for all $\ux \in (-\varepsilon_1, \varepsilon_1)^d$ and for all $i,j = 1, \ldots, d$.
\vspace*{2mm}

We compute
\begin{equation*}
\overline{g}_{ij} = \big(\omega^*\tilde{g}\big)_{ij} = \tilde{g}_{\mu \nu} \frac{\partial \omega^\mu}{\partial x_i} \frac{\partial \omega^\nu}{\partial x_j} = \tilde{g}_{00} \frac{\partial h}{\partial x_i} \frac{\partial h}{\partial x_j} + \tilde{g}_{0j}\frac{\partial h}{\partial x_i} + + \tilde{g}_{i0}\frac{\partial h}{\partial x_j} + \tilde{g}_{ij} \;.
\end{equation*}
It now follows from $|\tilde{g}_{\mu \nu}| \leq 1 + \delta_0$ and Step 2.2 that there exists a constant $0 < C_{\overline{g}} < \infty$ such that $|\overline{g}_{ij}| \leq C_{\overline{g}}$ holds for all  $i,j = 1, \ldots, d$.
\vspace*{2mm}

We can now finish Step 2. Let $p,q \in I^+\big(\gamma(-\mu), \Mint\big) \cap \Sigma$ be given. By property 1 of Step 1 there exists $\underline{x}, \underline{y} \in (-\varepsilon_1, \varepsilon_1)^d$ with $\omega(\underline{x}) = \tilde{\varphi} \big( \iota(p)\big)$ and  $\omega(\underline{y}) = \tilde{\varphi} \big( \iota(q)\big)$.
We connect the points $\underline{x}$ and $\underline{y}$ by a straight line in $(-\varepsilon_1, \varepsilon_1)^d$, i.e., by $\sigma : [0,1] \to (-\varepsilon_1, \varepsilon_1)^d$, where $\sigma(s) = \ux + s(\underline{y} - \ux)$. It follows that
\begin{equation*}
\begin{split}
L(\sigma) &= \int_0^1 \sqrt{\overline{g}\big( \dot{\sigma}(s), \dot{\sigma}(s)\big)} \, ds \\
&= \int_0^1 \sqrt{\Big(\sum_{i, j =1}^d (\underline{y}_i - \ux_i) \overline{g}_{ij}\big(\sigma(s)\big)(\underline{y}_j - \ux_j) \Big)}\,ds \\
&\leq \int_0^1 \sqrt{ \Big( \sum_{i, j =1}^d  \sqrt{d}\varepsilon_1 \cdot C_{\overline{g}} \cdot \sqrt{d}\varepsilon_1\Big)} \, ds \\
&= \varepsilon_1 d^{\nicefrac{3}{2}} \sqrt{C_{\overline{g}}} \;.
\end{split}
\end{equation*}
Note that this bound does not depend on the points $\ux, \underline{y} \in (-\varepsilon_1, \varepsilon_1)^d$. It follows that $\iota^{-1} \circ \tilde{\varphi}^{-1} \circ \omega \circ \sigma$ is a smooth curve in $\Sigma$ that connects $p$ and $q$ and the length of which is less or equal to $\varepsilon_1 d^{\nicefrac{3}{2}} \sqrt{C_{\overline{g}}}$. Since
\begin{equation*}
\begin{split}
d_{\Sigma}(p,q) &= \inf_{\substack{\gamma : [0,1] \to \Sigma \textnormal{ piecewise smooth} \\\gamma(0) = p \textnormal{ and } \gamma(1) = q }}   \Big{\{} \int_0^1 \sqrt{ \overline{g}\big(\dot{\gamma}(s), \dot{\gamma}(s)\big)} \, ds\Big{\}} \\ 
&\leq \inf_{\substack{\gamma : [0,1] \to S \textnormal{ piecewise smooth} \\\gamma(0) = p \textnormal{ and } \gamma(1) = q }}   \Big{\{} \int_0^1 \sqrt{ \overline{g}\big(\dot{\gamma}(s), \dot{\gamma}(s)\big)} \, ds\Big{\}}\\ 
&\leq \varepsilon_1 d^{\nicefrac{3}{2}} \sqrt{C_{\overline{g}}} \;,
\end{split}
\end{equation*}
this concludes Step 2 with $C_d = \varepsilon_1 d^{\nicefrac{3}{2}} \sqrt{C_{\overline{g}}}$.
\vspace*{2mm}

\underline{\textbf{Step 3:}} We show that the result of Step 2 is in contradiction with the geometry of $(\Mint, \gint)$.
\vspace*{2mm}

Let $\omega_0$ be the projection of $\gamma(-\frac{\mu}{2})$ to $\Sd$, $t_0 := t\big(\gamma(-\frac{\mu}{2})\big)$, and $r_0 := r\big(\gamma(-\frac{\mu}{2})\big)$. By the openness of $I^+\big(\gamma(-\mu), \Mint\big)$ there exists a $\lambda >0$ such that $(t_0 - \lambda, t_0 + \lambda) \times {r_0} \times \omega_0 \subseteq I^+\big(\gamma(-\mu), \Mint\big)$. Since $-\partial_r$ is future directed timelike, we also have
\begin{equation*}
[t_0 - \lambda, t_0 + \lambda] \times (0,r_0) \times \omega_0 \subseteq I^+\big(\gamma(-\mu), \Mint\big) \;.
\end{equation*}
Next, we consider the family of Cauchy hypersurfaces $\Sigma_n := \{r = \frac{1}{n}\}$, where $n \in \N_{>n_0}$ for some large $n_0 \in \N$, together with the family of pairs of points
\begin{equation*}
p_n := (t_0 - \lambda, \frac{1}{n}, \omega_0) \in I^+\big(\gamma(-\mu), \Mint\big) \qquad \textnormal{ and } \qquad q_n := (t_0 + \lambda, \frac{1}{n}, \omega_0) \in I^+\big(\gamma(-\mu), \Mint\big) \;.
\end{equation*}
The induced metric $\overline{g}_n$ on $\Sigma_n$ is given by
\begin{equation*}
\overline{g}_n = -\big(1 - 2m\cdot n^{(d-2)}\big) \,dt^2 + n^{-2}\, \mathring{\gamma}_{d-1} \;.
\end{equation*}
It is easy to see that the shortest curve connecting $p_n$ with $q_n$ in $\Sigma_n$ is given by
 $\gamma_n : (-\lambda, \lambda) \to \Sigma_n$, 
\begin{equation*}
\gamma_n (s) = ( t_0 + s, \frac{1}{n}, \omega_0) \;.
\end{equation*}
The length $L(\gamma_n)$ of $\gamma_n$ is given by
\begin{equation*}
\begin{split}
L(\gamma_n) &= \int_{-\lambda}^\lambda \sqrt{\overline{g}_n \big(\dot{\gamma}_n(s), \dot{\gamma}_n(s)\big)} \, ds \\
&= \int_{-\lambda}^\lambda \sqrt{2m \cdot n^{(d-2)} -1} \, ds \\ 
&= 2 \lambda \sqrt{2m \cdot n^{(d-2)} -1} \;.
\end{split}
\end{equation*}
Since $d \geq 3$, it thus follows that
\begin{equation*}
d_{\Sigma_n}(p_n, q_n)= 2\lambda \sqrt{2m \cdot n^{(d-2)} -1} \to \infty \quad \textnormal{ for } n \to \infty \;,
\end{equation*}
in contradiction to Step 2.
This concludes the proof of Theorem \ref{MainThm}.
\end{proof}

\bibliographystyle{acm}
\bibliography{Bibly}

\end{document}

%% file: BoundaryExamples.pdf_tex
\begingroup%
  \makeatletter%
  \providecommand\color[2][]{%
    \errmessage{(Inkscape) Color is used for the text in Inkscape, but the package 'color.sty' is not loaded}%
    \renewcommand\color[2][]{}%
  }%
  \providecommand\transparent[1]{%
    \errmessage{(Inkscape) Transparency is used (non-zero) for the text in Inkscape, but the package 'transparent.sty' is not loaded}%
    \renewcommand\transparent[1]{}%
  }%
  \providecommand\rotatebox[2]{#2}%
  \ifx\svgwidth\undefined%
    \setlength{\unitlength}{530.48603516bp}%
    \ifx\svgscale\undefined%
      \relax%
    \else%
      \setlength{\unitlength}{\unitlength * \real{\svgscale}}%
    \fi%
  \else%
    \setlength{\unitlength}{\svgwidth}%
  \fi%
  \global\let\svgwidth\undefined%
  \global\let\svgscale\undefined%
  \makeatother%
  \begin{picture}(1,0.60399463)%
    \put(0,0){\includegraphics[width=\unitlength]{BoundaryExamples.pdf}}%
    \put(0.07183836,0.06389684){\color[rgb]{0,0,0}\makebox(0,0)[lb]{\smash{$\tilde{M}=\mathbb{R} \times \mathbb{R}$}}}%
    \put(0.07183836,0.01003795){\color[rgb]{0,0,0}\makebox(0,0)[lb]{\smash{$\tilde{g} = - dt^2 + dx^2$}}}%
    \put(0.59534753,0.06605131){\color[rgb]{0,0,0}\makebox(0,0)[lb]{\smash{$\tilde{M}=\mathbb{R} \times S^1$}}}%
    \put(0.59534748,0.01219224){\color[rgb]{0,0,0}\makebox(0,0)[lb]{\smash{$\tilde{g}= -dt^2 + d\varphi^2$}}}%
    \put(0.08186563,0.16157692){\color[rgb]{0,0,0}\makebox(0,0)[lb]{\smash{$p$}}}%
    \put(0.08186563,0.55582452){\color[rgb]{0,0,0}\makebox(0,0)[lb]{\smash{$q$}}}%
    \put(0.32099941,0.38778457){\color[rgb]{0,0,0}\makebox(0,0)[lb]{\smash{$r$}}}%
    \put(0.64199885,0.4674958){\color[rgb]{0,0,0}\makebox(0,0)[lb]{\smash{$s$}}}%
  \end{picture}%
\endgroup%

%% file: BoundaryAchronal.pdf_tex
\begingroup%
  \makeatletter%
  \providecommand\color[2][]{%
    \errmessage{(Inkscape) Color is used for the text in Inkscape, but the package 'color.sty' is not loaded}%
    \renewcommand\color[2][]{}%
  }%
  \providecommand\transparent[1]{%
    \errmessage{(Inkscape) Transparency is used (non-zero) for the text in Inkscape, but the package 'transparent.sty' is not loaded}%
    \renewcommand\transparent[1]{}%
  }%
  \providecommand\rotatebox[2]{#2}%
  \ifx\svgwidth\undefined%
    \setlength{\unitlength}{331.03303223bp}%
    \ifx\svgscale\undefined%
      \relax%
    \else%
      \setlength{\unitlength}{\unitlength * \real{\svgscale}}%
    \fi%
  \else%
    \setlength{\unitlength}{\svgwidth}%
  \fi%
  \global\let\svgwidth\undefined%
  \global\let\svgscale\undefined%
  \makeatother%
  \begin{picture}(1,1.18658853)%
    \put(0,0){\includegraphics[width=\unitlength]{BoundaryAchronal.pdf}}%
    \put(0.76418792,0.58466331){\color[rgb]{0,0,0}\makebox(0,0)[lb]{\smash{$\varepsilon_1$}}}%
    \put(0.43275783,1.11978467){\color[rgb]{0,0,0}\makebox(0,0)[lb]{\smash{$\varepsilon_0$}}}%
    \put(0.58466335,0.51561529){\color[rgb]{0,0,0}\makebox(0,0)[lb]{\smash{$\underline{x}_0$}}}%
    \put(0.47418658,0.87811702){\color[rgb]{0,0,0}\makebox(0,0)[lb]{\smash{$\sigma$}}}%
    \put(0.82978344,0.84014057){\color[rgb]{0,0,0}\makebox(0,0)[lb]{\smash{$\rho_\tau$}}}%
    \put(0.81191457,0.74392933){\color[rgb]{0,0,0}\makebox(0,0)[lb]{\smash{$\rho_{\tau_0}$}}}%
    \put(0.21525688,0.61573489){\color[rgb]{0,0,0}\makebox(0,0)[lb]{\smash{$q$}}}%
  \end{picture}%
\endgroup%

%% file: Cones.pdf_tex
\begingroup%
  \makeatletter%
  \providecommand\color[2][]{%
    \errmessage{(Inkscape) Color is used for the text in Inkscape, but the package 'color.sty' is not loaded}%
    \renewcommand\color[2][]{}%
  }%
  \providecommand\transparent[1]{%
    \errmessage{(Inkscape) Transparency is used (non-zero) for the text in Inkscape, but the package 'transparent.sty' is not loaded}%
    \renewcommand\transparent[1]{}%
  }%
  \providecommand\rotatebox[2]{#2}%
  \ifx\svgwidth\undefined%
    \setlength{\unitlength}{387.66550293bp}%
    \ifx\svgscale\undefined%
      \relax%
    \else%
      \setlength{\unitlength}{\unitlength * \real{\svgscale}}%
    \fi%
  \else%
    \setlength{\unitlength}{\svgwidth}%
  \fi%
  \global\let\svgwidth\undefined%
  \global\let\svgscale\undefined%
  \makeatother%
  \begin{picture}(1,1.1572284)%
    \put(0,0){\includegraphics[width=\unitlength]{Cones.pdf}}%
    \put(0.67120495,1.03853578){\color[rgb]{0,0,0}\makebox(0,0)[lb]{\smash{$(-\varepsilon_0, \varepsilon_0)\times (-\varepsilon_1, \varepsilon_1)^d$}}}%
    \put(0.46731969,1.10705629){\color[rgb]{0,0,0}\makebox(0,0)[lb]{\smash{$x_0$}}}%
    \put(0.91351218,0.48939639){\color[rgb]{0,0,0}\makebox(0,0)[lb]{\smash{$\underline{x}$}}}%
    \put(0.45729656,0.41389839){\color[rgb]{0,0,0}\makebox(0,0)[lb]{\smash{$y^-$}}}%
    \put(0.4593811,0.16374792){\color[rgb]{0,0,0}\makebox(0,0)[lb]{\smash{$x^-$}}}%
    \put(0.46355032,0.91002968){\color[rgb]{0,0,0}\makebox(0,0)[lb]{\smash{$x^+$}}}%
    \put(0.69285479,0.86625339){\color[rgb]{0,0,0}\makebox(0,0)[lb]{\smash{$x^+ + C^-_{\nicefrac{5}{6}}$}}}%
    \put(0.70119315,0.69323278){\color[rgb]{0,0,0}\makebox(0,0)[lb]{\smash{$x^+ + C^-_{\nicefrac{6}{7}}$}}}%
    \put(0.68243182,0.24087766){\color[rgb]{0,0,0}\makebox(0,0)[lb]{\smash{$x^- + C^+_{\nicefrac{5}{6}}$}}}%
    \put(0.68868551,0.35344526){\color[rgb]{0,0,0}\makebox(0,0)[lb]{\smash{$x^- + C^+_{\nicefrac{6}{7}}$}}}%
    \put(-0.00548153,0.28465395){\color[rgb]{0,0,0}\makebox(0,0)[lb]{\smash{$C^-_{\nicefrac{5}{8}} \cap \big(y^- + C^+_{\nicefrac{5}{8}}\big)$}}}%
  \end{picture}%
\endgroup%

%% file: Separation.pdf_tex
\begingroup%
  \makeatletter%
  \providecommand\color[2][]{%
    \errmessage{(Inkscape) Color is used for the text in Inkscape, but the package 'color.sty' is not loaded}%
    \renewcommand\color[2][]{}%
  }%
  \providecommand\transparent[1]{%
    \errmessage{(Inkscape) Transparency is used (non-zero) for the text in Inkscape, but the package 'transparent.sty' is not loaded}%
    \renewcommand\transparent[1]{}%
  }%
  \providecommand\rotatebox[2]{#2}%
  \ifx\svgwidth\undefined%
    \setlength{\unitlength}{461.55961914bp}%
    \ifx\svgscale\undefined%
      \relax%
    \else%
      \setlength{\unitlength}{\unitlength * \real{\svgscale}}%
    \fi%
  \else%
    \setlength{\unitlength}{\svgwidth}%
  \fi%
  \global\let\svgwidth\undefined%
  \global\let\svgscale\undefined%
  \makeatother%
  \begin{picture}(1,0.87301212)%
    \put(0,0){\includegraphics[width=\unitlength]{Separation.pdf}}%
    \put(0.15242731,0.50898242){\color[rgb]{0,0,0}\makebox(0,0)[lb]{\smash{$K$}}}%
    \put(0.44955655,0.52879096){\color[rgb]{0,0,0}\makebox(0,0)[lb]{\smash{$\gamma$}}}%
    \put(0.20937706,0.20690094){\color[rgb]{0,0,0}\makebox(0,0)[lb]{\smash{$\sigma$}}}%
    \put(0.53621919,0.09052538){\color[rgb]{0,0,0}\makebox(0,0)[lb]{\smash{$I^-\big(\gamma(y_0^-),M_\mathrm{int}\big)$}}}%
    \put(0.63773837,0.36289379){\color[rgb]{0,0,0}\makebox(0,0)[lb]{\smash{$\gamma(y_0^+)$}}}%
    \put(0.60802545,0.20937715){\color[rgb]{0,0,0}\makebox(0,0)[lb]{\smash{$\gamma(y_0^-)$}}}%
    \put(0.47926947,0.83087239){\color[rgb]{0,0,0}\makebox(0,0)[lb]{\smash{$\{r=0\}$}}}%
  \end{picture}%
\endgroup%

%% file: Cones2.pdf_tex
\begingroup%
  \makeatletter%
  \providecommand\color[2][]{%
    \errmessage{(Inkscape) Color is used for the text in Inkscape, but the package 'color.sty' is not loaded}%
    \renewcommand\color[2][]{}%
  }%
  \providecommand\transparent[1]{%
    \errmessage{(Inkscape) Transparency is used (non-zero) for the text in Inkscape, but the package 'transparent.sty' is not loaded}%
    \renewcommand\transparent[1]{}%
  }%
  \providecommand\rotatebox[2]{#2}%
  \ifx\svgwidth\undefined%
    \setlength{\unitlength}{464.84169922bp}%
    \ifx\svgscale\undefined%
      \relax%
    \else%
      \setlength{\unitlength}{\unitlength * \real{\svgscale}}%
    \fi%
  \else%
    \setlength{\unitlength}{\svgwidth}%
  \fi%
  \global\let\svgwidth\undefined%
  \global\let\svgscale\undefined%
  \makeatother%
  \begin{picture}(1,0.96509743)%
    \put(0,0){\includegraphics[width=\unitlength]{Cones2.pdf}}%
    \put(0.64006567,0.86382714){\color[rgb]{0,0,0}\makebox(0,0)[lb]{\smash{$(-\varepsilon_0, \varepsilon_0)\times(-\varepsilon_1, \varepsilon_1)^d$}}}%
    \put(0.45579464,0.92325523){\color[rgb]{0,0,0}\makebox(0,0)[lb]{\smash{$x_0$}}}%
    \put(0.74793668,0.40988201){\color[rgb]{0,0,0}\makebox(0,0)[lb]{\smash{$\underline{x}$}}}%
    \put(0.44743562,0.34518015){\color[rgb]{0,0,0}\makebox(0,0)[lb]{\smash{$y^-$}}}%
    \put(0.44917407,0.13656137){\color[rgb]{0,0,0}\makebox(0,0)[lb]{\smash{$x^-$}}}%
    \put(0.45265108,0.75894033){\color[rgb]{0,0,0}\makebox(0,0)[lb]{\smash{$x^+$}}}%
    \put(0.63345393,0.72069357){\color[rgb]{0,0,0}\makebox(0,0)[lb]{\smash{$x^+ + C^-_{\nicefrac{5}{6}}$}}}%
    \put(0.6404079,0.57987599){\color[rgb]{0,0,0}\makebox(0,0)[lb]{\smash{$x^+ + C^-_{\nicefrac{6}{7}}$}}}%
    \put(0.65257723,0.21827035){\color[rgb]{0,0,0}\makebox(0,0)[lb]{\smash{$x^- + C^+_{\nicefrac{5}{6}}$}}}%
    \put(0.64910032,0.29997932){\color[rgb]{0,0,0}\makebox(0,0)[lb]{\smash{$x^- + C^+_{\nicefrac{6}{7}}$}}}%
    \put(-0.00457144,0.19630281){\color[rgb]{0,0,0}\makebox(0,0)[lb]{\smash{$C^-_{\nicefrac{5}{8}} \cap \big(y^- + C^+_{\nicefrac{5}{8}}\big)$}}}%
    \put(0.77855877,0.496847){\color[rgb]{0,0,0}\makebox(0,0)[lb]{\smash{$(\tilde{\varphi} \circ \iota)\big(\gamma(-\mu)\big)$}}}%
    \put(0.69781073,0.3623443){\color[rgb]{0,0,0}\makebox(0,0)[lb]{\smash{$y^+$}}}%
    \put(0.07228827,0.57376845){\color[rgb]{0,0,0}\makebox(0,0)[lb]{\smash{$\tilde{\varphi}\circ \iota \circ \sigma$}}}%
    \put(0.17330156,0.4941774){\color[rgb]{0,0,0}\makebox(0,0)[lb]{\smash{$\tau$}}}%
  \end{picture}%
\endgroup%